%% file: arXiv.tex
  \documentclass[10pt]{article}
  \input{gtarXiv.txt}
  \usepackage{amsmath,amsthm,amsfonts,amssymb,latexsym,graphicx,mathtools,euscript,url}
  \newcommand{\Extra}[1]{}

\allowdisplaybreaks[4]

\newcommand*{\st}{\mathrel{|}}
\newcommand*{\dd}{\,\mathrm{d}}

\newcommand*{\K}{\mathcal{K}}

\newcommand*{\FFF}{\mathcal{F}}
\newcommand{\Wiener}{\EuScript{W}}

\newcommand*{\Alpha}{\mathrm{A}}
\newcommand*{\Mu}{\mathrm{M}}

\DeclareMathOperator{\III}{\boldsymbol{1}}

\newcommand*{\bbbp}{\mathbb{P}}
\DeclareMathOperator{\Prob}{\bbbp}
\DeclareMathOperator{\UpProb}{\overline{\bbbp}}

\newcommand{\pretend}[2]{\smash{\mathrlap{#1}}\phantom{#2}}
\DeclareMathOperator{\UpProbI}{{\textstyle\pretend{\overline{\bbbp}}{\bbbp}^{\it I}}}
\DeclareMathOperator{\LowProbI}{{\textstyle\pretend{\underline{\bbbp}}{\bbbp}^{\it I}}}

\newcommand{\bbbe}{\mathbb{E}}
\DeclareMathOperator{\Expect}{\bbbe}
\DeclareMathOperator{\UpExpect}{\overline{\bbbe}}

\DeclareMathOperator{\UpExpectI}{\textstyle\pretend{\overline{\bbbe}}{\bbbe}^{\it I}}

\newcommand*{\bbbr}{\mathbb{R}}

\theoremstyle{plain}
\newtheorem{theorem}{Theorem}[section]
\newtheorem{proposition}[theorem]{Proposition}
\newtheorem{corollary}[theorem]{Corollary}
\newtheorem{lemma}[theorem]{Lemma}
\theoremstyle{definition}
\newtheorem{remark}[theorem]{Remark}

  \title{A probability-free and continuous-time explanation of the equity premium and CAPM}
  \author{Vladimir Vovk and Glenn Shafer}
  \date{July 4, 2016}
  \newcommand*{\No}{44}
  \twodatestrue
  \newcommand*{\firstposted}{June 23, 2016}

\begin{document}
\maketitle

  \begin{abstract}
    This paper gives yet another definition of game-theoretic probability
    in the context of continuous-time idealized financial markets.
    Without making any probabilistic assumptions
    (but assuming positive and continuous price paths),
    we obtain a simple expression for the equity premium
    and derive a version of the capital asset pricing model.

      \bigskip

      \bigskip

      \noindent
      The version of this paper at \url{http://probabilityandfinance.com} (Working Paper 44)
      is updated most often.
  \end{abstract}

\section{Introduction}

This paper reviews and extends previous work
in which we derived the existence of an equity premium
and the validity of a Capital Asset Pricing Model (CAPM)
from a purely game-theoretic hypothesis of market efficiency,
without assuming the existence of probabilities for security prices.

For simplicity, we consider only two securities,
a stock $S$ and a traded market index $I$.
We also make the following simplifying assumptions:
\begin{itemize}
\item
  Trading in $S$ and $I$ continues indefinitely.
  (The time horizon is infinite.)
\item
  The prices of $S$ and $I$ are always positive and continuous.  
\item
  The interest rate is zero.
\end{itemize}
All these assumptions can be relaxed.

Our mathematical results have a practical interpretation if one adopts
the hypothesis that the index $I$ is efficient, in the sense that a 
strategy for trading in $S$ and $I$ will not multiply the 
capital it risks by a factor many times larger than what would be achieved
by buying and holding $I$.  We call this the \emph{Efficient Index Hypothesis
(EIH) for $I$}.  

A typical mathematical result in this paper asserts the existence of a trading
strategy that will multiply the 
capital it risks by a factor many times larger than what would be achieved
by buying and holding $I$ unless the price trajectories of
$S$ and $I$ have a particular property.  
Here are some properties we consider:
\begin{itemize}
\item
  $I$ grows at a rate determined by 
  its volatility.  (This is the equity premium.)
\item
  $I$ has the properties of 
  geometric Brownian motion when time is appropriately rescaled.
\item
  $S$ obeys a CAPM with respect to $I$.  
\end{itemize}
In each case, we prove the existence of a trading strategy that beats $I$
by a large factor if the property does not hold.
If you subscribe to the EIH for $I$, then you expect the property to hold.

Our EIH is explained more fully in Section~\ref{sec:EIH},
where we specify the trading strategies we consider,
define an extended class of approximate capital processes (supermartingales),
and state the associated definition of upper probability.
An upper probability measures how little initial capital must be risked
to obtain unit capital if an event happens and thus how unlikely that event is.

We study the index $I$ in Sections~\ref{sec:existence-1}--\ref{sec:EPM}.
In Section~\ref{sec:existence-1}
we define $I$'s cumulative growth rate and relative quadratic variation;
these exist in a strong sense under the EIH:
the trader can become infinitely rich as soon as they cease to exist.
In Sections~\ref{sec:EP} and~\ref{sec:EPM}
we consider strategies for trading in $I$
and show that under the EIH it grows at a rate determined by its relative quadratic variation.
This growth is the equity premium.

Section~\ref{sec:existence-2} continues Section~\ref{sec:existence-1}
by defining quantities involving both $S$ and $I$.
Section~\ref{sec:CAPM} then derives a CAPM that relates these quantities to each other
and includes most of our results about the equity premium as special cases.

One purpose of this paper is to clarify the relation between two different methods
that we used in previous work.
We first established a probability-free CAPM fifteen years ago;
essentially, our version was the conjunction of a probability-free version of the standard CAPM
and a probability-free expression for the equity premium.
We did this first in discrete time \cite{GTP1},
and then we extended the argument to continuous time using nonstandard analysis \cite{GTP2}.
The method used in those papers involved mixing, in a certain sense,
the price paths of $S$ and $I$ (in the case of CAPM)
or mixing $I$ and cash (in the case of the equity premium).
Ten years later, without using nonstandard analysis,
one of us derived a probability-free version of the Dubins--Schwarz theorem \cite{GTP28},
effectively reducing the probability-free setting to the Bachelier model,
which for positive prices becomes the Black--Scholes model after a time change.
And the Black--Scholes model allows us to use standard probabilistic tools,
including Girsanov's theorem,
to obtain a stronger form of our version of the CAPM \cite{GTP38,GTP39-local,Vovk:arXiv1111}.

In this paper we apply and compare the two methods,
mixing \cite{GTP1,GTP2} and probabilistic \cite{GTP38,GTP39-local,Vovk:arXiv1111},
implementing them both without using nonstandard analysis.
In Section~\ref{sec:EPM}, we study the equity premium using the mixing method,
and in Section~\ref{sec:EP}, we study it 
using the probabilistic method in combination with the probability-free Dubins--Schwarz theorem.
The results from the probabilistic method are stronger than those from the mixing method,
in the sense that they assert higher lower probabilities
for the approximations formalizing the equity premium phenomenon,
but the difference is not great.
Since the probability-free Dubins--Schwarz theorem is only applicable to one security,
we cannot apply the probabilistic method to the CAPM,
which involves both $S$ and $I$.
Therefore, we use the mixing method to obtain our version of the CAPM in Section~\ref{sec:CAPM}.

\section{The Efficient Index Hypothesis}
\label{sec:EIH}

The sample space of this paper is the set $\Omega$
of all pairs $\omega=(I,S)$ of positive continuous functions
$I:[0,\infty)\to(0,\infty)$ and $S:[0,\infty)\to(0,\infty)$.
Each $\omega=(I,S)\in\Omega$ will be identified with the function
$\omega:[0,\infty)\to(0,\infty)^2$
defined by $\omega(t):=(I(t),S(t))$, $t\in[0,\infty)$.
Intuitively, $I$ is the price path of an index and $S$ is that 
of a stock or another financial security.
We assume, for simplicity, that $I(0)=1$.

We equip $\Omega$ with the $\sigma$-algebra $\FFF$
generated by the functions $\omega\in\Omega\mapsto\omega(t)$, $t\in[0,\infty)$
(i.e., the smallest $\sigma$-algebra making them measurable).
We often consider subsets of $\Omega$ and functions on $\Omega$
that are measurable with respect to $\FFF$.
As shown in \cite{GTP43}, the requirement of measurability is essential:
without measurability, it is too easy to become infinitely rich infinitely quickly.

An \emph{event} is an arbitrary subset of $\Omega$
(we will add the qualifier ``$\FFF$-measurable'' when needed),
a \emph{random vector} is an $\FFF$-measurable function of the type $\Omega\to\bbbr^d$ for some $d\in\{1,2,\ldots\}$,
and an \emph{extended random variable} is an $\FFF$-measurable function of the type $\Omega\to[-\infty,\infty]$.
A \emph{stopping time} is an extended random variable $\tau:\Omega\to[0,\infty]$ such that,
for all $\omega$ and $\omega'$ in $\Omega$,
\begin{equation*}
  \left(
    \omega|_{[0,\tau(\omega)]}
    =
    \omega'|_{[0,\tau(\omega)]}
  \right)
  \Longrightarrow
  \tau(\omega)=\tau(\omega'),
\end{equation*}
where $f|_A$ stands for the restriction of $f$ to the intersection of $A$ and $f$'s domain.
A random vector $X$ is said to be \emph{$\tau$-measurable},
where $\tau$ is a stopping time,
if,
for all $\omega$ and $\omega'$ in $\Omega$,
\begin{equation*}
  \left(
    \omega|_{[0,\tau(\omega)]}
    =
    \omega'|_{[0,\tau(\omega)]}
  \right)
  \Longrightarrow
  X(\omega)=X(\omega').
\end{equation*}
As customary in probability theory,
we will often omit explicit mention of $\omega\in\Omega$
when it is clear from the context.

A \emph{simple trading strategy} $G$ is a pair $((\tau_1,\tau_2,\ldots),(h_1,h_2,\ldots))$,
where:
\begin{itemize}
\item
  $\tau_1\le\tau_2\le\cdots$ is a nondecreasing sequence of stopping times
  such that, for each $\omega\in\Omega$,
  $\lim_{n\to\infty}\tau_n(\omega)=\infty$;
\item
  for each $n=1,2,\ldots$, $h_n$ is a bounded $\tau_{n}$-measurable $\bbbr^2$-valued random vector.
\end{itemize}
A \emph{process} is a function $X:[0,\infty)\times\Omega\to[-\infty,\infty]$.
The \emph{simple capital process} $\K^{G,c}$
corresponding to a simple trading strategy $G$
and \emph{initial capital} $c\in\bbbr$ is defined by
\begin{equation*}
  \K^{G,c}_t(\omega)
  :=
  c
  +
  \sum_{n=1}^{\infty}
  h_n(\omega)\cdot
  \bigl(
    \omega(\tau_{n+1}\wedge t)-\omega(\tau_n\wedge t)
  \bigr),
  \quad
  t\in[0,\infty),
  \enspace
  \omega\in\Omega,
\end{equation*}
where ``$\cdot$'' stands for dot product and the zero terms in the sum are ignored
(which makes the sum finite for each $t$).

The vector $h_n(\omega)$ tells the trader how many units of $I$ and $S$ to 
hold between time $\tau_n(\omega)$ and 
$\tau_{n+1}(\omega)$, and 
thus $\K^{G,c}_t(\omega)$ is his \emph{capital} at time $t$.
Negative components for $h_n$ indicate short selling.  
Because $I$ and $S$ are continuous, a strategy $G$ can sell
them short and yet produce a nonnegative simple capital 
process $\K^{G,c}$; the $\tau_n$ and $h_n$ can be chosen so that
the short selling always stops before
$\K^{G,c}_t(\omega)$ gets below zero.

For $\omega=(I,S)$ and $t\in[0,\infty)$, we often let $I_t(\omega)$ stand for $I(t)$
and $S_t(\omega)$ for $S(t)$;
since we often omit $\omega$,
this makes $I_t$ (resp.\ $S_t$) synonymous with $I(t)$ (resp.\ $S(t)$).
We will often use the generic notation $\omega=(I,S)$.

Let us say that a class $\mathcal{C}$ of processes is \emph{$\liminf$-closed}
if the process
\begin{equation}\label{eq:PP}
  X_t(\omega)
  :=
  \liminf_{k\to\infty}
  X^k_t(\omega)
\end{equation}
is in $\mathcal{C}$ whenever each process $X^k$ is in $\mathcal{C}$.
A nonnegative process $X$
is a \emph{test supermartingale} if it belongs to the smallest $\liminf$-closed class of processes
containing all nonnegative simple capital processes.
Intuitively, test supermartingales are nonnegative capital processes
(as they can be approximated by nonnegative simple capital processes;
in fact, they can lose capital as the approximation is in the sense of $\liminf$).
We call processes of the type $X_t(\omega) / I_t(\omega)$, 
where $X$ is a test supermartingale,
\emph{test $I$-supermartingales};
they are like test supermartingales but use $I$ as the num\'eraire.

The initial value $X_0$ of a test supermartingale $X$ is always a 
constant.  Given a subset
$E$ of $\Omega$, we set
\begin{equation}\label{eq:upper-probability}
  \UpProb(E)
  :=
  \inf
  \bigl\{
    X_0
    \bigm|
    \forall\omega\in\Omega:
    \liminf_{t\to\infty}
    X_t(\omega)
    \ge
    \III_E(\omega)
  \bigr\}
\end{equation}
and
\begin{equation}\label{eq:upper-probability-I}
  \UpProbI(E)
  :=
  \inf
  \bigl\{
    X_0
    \bigm|
    \forall\omega\in\Omega:
    \liminf_{t\to\infty}
    X_t(\omega) / I_t(\omega)
    \ge
    \III_E(\omega)
  \bigr\},
\end{equation}
$X$  ranging in each case over the test supermartingales.%
\footnote{Here, as always in game-theoretic probability,
  upper probability is a special case of upper expected value.
  Upper expected values $\UpExpect(F)$  and $\UpExpectI(F)$,
  where $F:\Omega\to[0,\infty]$,
  are defined by substituting the function $F$
  for $\III_E$ in~\eqref{eq:upper-probability} and \eqref{eq:upper-probability-I}, respectively.
  We do not use $\UpExpectI$ in this paper but do use $\UpExpect$ on one occasion.}            
We call $\UpProb(E)$ $E$'s \emph{upper probability},
and we call $\UpProbI(E)$ its $I$-\emph{upper probability}.
The definition~\eqref{eq:upper-probability-I} can be rewritten as
\begin{equation}\label{eq:upper-probability-II}
  \UpProbI(E)
  =
  \inf
  \bigl\{
    X_0
    \bigm|
    \forall\omega\in\Omega:
    \liminf_{t\to\infty}
    X_t(\omega)
    \ge
    \III_E(\omega)
  \bigr\},
\end{equation}
$X$ ranging over the test $I$-supermartingales.

Recalling that $I_0(\omega)=1$ for all $\omega\in\Omega$,
we see from~\eqref{eq:upper-probability-I} that a value of $\UpProbI(E)$ close to zero
indicates the existence of a trading strategy that beats the index $I$
by a large factor if $E$ happens.
The EIH for $I$ says that we should not expect to beat $I$ by a large factor,
and so we should not expect $E$ to happen.
The EIH for $E$ has a lot of empirical support
when $I$ is an index, such as the S\&P500,
which can be approximately traded with low transaction costs;
see, e.g., \cite{Malkiel:1995,Malkiel:2016}.

We do not interpret small values of $\UpProb(E)$ in the same way.
Saying that $E$ will not happen when $\UpProb(E)$ is small
would amount to adopting an efficiency hypothesis for cash or for a bank account
(recall that the interest rate is zero)---%
i.e., to asserting that no trading strategy will beat holding cash by a large factor.
We do not assert this.
In fact, the EIH for nontrivial $I$ implies the opposite.
It implies that we can expect holding $I$ to beat holding cash by an infinite factor
as time goes to infinity;
this is a consequence of our results for the equity premium in Sections~\ref{sec:EP} and~\ref{sec:EPM}.
(The efficiency hypothesis for cash, in contrast,
implies that the price of $I$, or any other traded security, will tend to a constant.
See Theorem~3.1 in \cite{GTP28}.)

\begin{remark}\label{rem:supermartingales}
  An equivalent definition of the class $\mathcal{C}$ of test supermartingales
  can be given using transfinite induction over the countable ordinals $\alpha$
  (see, e.g., \cite{Dellacherie/Meyer:1978}, 0.8).
  Namely, define $\mathcal{C}^{\alpha}$ as follows:
  \begin{itemize}
  \item
    $\mathcal{C}^0$ is the class of all nonnegative simple capital processes;
  \item
    for $\alpha>0$,
    $X\in\mathcal{C}^{\alpha}$ if and only if there exists a sequence $X^1,X^2,\ldots$
    of processes in $\mathcal{C}^{<\alpha}:=\cup_{\beta<\alpha}\mathcal{C}^{\beta}$
    such that \eqref{eq:PP} holds.
  \end{itemize}
  It is easy to check that the class of all test supermartingales
  is the union of the nested family $\mathcal{C}^{\alpha}$
  over all countable ordinals $\alpha$.
  The \emph{class} of a test supermartingale $X$
  is defined to be the smallest $\alpha$ such that $X\in\mathcal{C}^{\alpha}$;
  in this case we will also say that $X$ is \emph{of class $\alpha$}.
\end{remark}

\begin{remark}
  The hierarchy $(\mathcal{C}^{\alpha})$ described in Remark~\ref{rem:supermartingales}
  is somewhat analogous to the Baire hierarchy of Borel functions on a metric space:
  see, e.g., \cite{Kechris:1995}, Section~24.
\end{remark}

\begin{remark}
  Our definition of upper probability is similar
  to the one given by Perkowski and Pr\"omel \cite{Perkowski/Promel:2016}
  (who modified the definition given in \cite{GTP28}).
  The main differences are that Perkowski and Pr\"omel
  define the upper probability~\eqref{eq:upper-probability}
  using the test supermartingales in the class $\mathcal{C}^1$ rather than $\mathcal{C}$
  (in the notation of Remark~\ref{rem:supermartingales})
  and that they consider a finite horizon
  (our time interval is $[0,\infty)$ instead of their $[0,T]$).
  The proofs of our results given below work for any $\mathcal{C}^n$, $n\ge2$,
  in place of $\mathcal{C}$.
\end{remark}

\begin{remark}\label{rem:measure-theoretic}
  The motivation for our terminology is the analogy with measure-theoretic probability.
  Namely, let us suppose that $I$ and $S$ are local martingales
  on a measure-theoretic probability space.
  Each simple capital process is a local martingale.
  Since each nonnegative local martingale is a supermartingale
  (\cite{Revuz/Yor:1999}, p.~123),
  nonnegative simple capital processes are supermartingales.
  By Fatou's lemma, $\liminf_k X^k$ is a supermartingale
  whenever $X^k$ are nonnegative supermartingales:
  \[
    \Expect
    \left(
      \liminf_k X^k_t\st\FFF_s
    \right)
    \le
    \liminf_k
    \Expect(X^k_t\st\FFF_s)
    \le
    \liminf_k X^k_s,
  \]
  where $0\le s<t$.
  Therefore, our definition gives a subset of the set of all nonnegative measure-theoretic supermartingales.
\end{remark}

\begin{remark}
  Let us check that, in the measure-theoretic setting of Remark~\ref{rem:measure-theoretic}
  (where $I$ and $S$ are local martingales),
  $\UpProb(E)\ge\Prob(E)$ for each $\FFF$-measurable $E$.
  (In this sense our definition \eqref{eq:upper-probability} of $\UpProb$
  is not too permissive, unlike the definition ignoring measurability in \cite{GTP43}.)
  It suffices to establish the ``maximal inequality''
  for nonnegative measure-theoretic supermartingales $X$ with $X_0$ a constant
  in the form
  \begin{equation*}
    \Prob
    \left(
      \liminf_{t\to\infty}X_t\ge1
    \right)
    \le
    X_0.
  \end{equation*}
  To check this, notice that, for each $\epsilon\in(0,1)$,
  \begin{align*}
    \Prob
    \left(
      \liminf_{t\to\infty}X_t\ge1
    \right)
    &\le
    \Prob(X_t\ge1-\epsilon \text{ from some $t$ on})\\
    &\le
    \liminf_{T\to\infty}
    \Prob(X_t\ge1-\epsilon \text{ for all $t\ge T$})\\
    &\le
    \liminf_{T\to\infty}
    \Prob(X_T\ge1-\epsilon)
    \le
    \frac{X_0}{1-\epsilon}.
  \end{align*}
\end{remark}

\begin{remark}\label{rem:I-supermartingales}
  Let us say that $X$ is a class $\alpha$ test $I$-supermartingale,
  where $\alpha$ is a countable ordinal,
  if $XI$ is a class $\alpha$ test supermartingale, as defined in Remark~\ref{rem:supermartingales}.
  For each countable ordinal $\alpha$ and each class $\alpha$ test $I$-supermartingale
  we fix a sequence $X^k$ of test $I$-supermartingales of smaller classes
  such that $X=\liminf_{k\to\infty}X^k$
  (as usual, we are using the axiom of choice freely).
\end{remark}

The following lemma says that the definition~\eqref{eq:upper-probability-II}
is robust in that the $\liminf$ in it can be replaced by $\limsup$ or even $\sup$.

\begin{lemma}\label{lem:robustness}
  For any $E\subseteq\Omega$,
  \begin{equation}\label{eq:upper-probability-III}
    \UpProbI(E)
    =
    \inf
    \bigl\{
      X_0
      \bigm|
      \forall\omega\in\Omega:
      \sup_{t\in[0,\infty)}
      X_t(\omega)
      \ge
      \III_E(\omega)
    \bigr\},
  \end{equation}
  $X$ ranging over the test $I$-supermartingales.
\end{lemma}

\begin{proof}
  The only nontrivial part of the equality in~\eqref{eq:upper-probability-III} is the inequality~``$\le$'',
  and it is clear that we can replace ``${}\ge\III_E(\omega)$'' by ``${}>\III_E(\omega)$''.
  This is what we will be proving.

  For each test $I$-supermartingale $X$ with $X_0<1$
  we will define another test $I$-supermartingale $X^*$, satisfying $X^*_0=X_0$,
  as follows.
  If $X$ is a simple capital process, set
  \[
    X^*_t
    :=
    \begin{cases}
      X_t & \text{if $\sup_{s\in[0,t]}X_s<1$}\\
      1 & \text{otherwise}.
    \end{cases}
  \]
  If $X$ is a class $\alpha$ test $I$-supermartingale, we set
  $X^*=\liminf_{k\to\infty}(X^k)^*$,
  where $X^k$ is the fixed sequence of $I$-supermartingales of classes smaller than that of $X$
  (see Remark~\ref{rem:I-supermartingales}).

  It suffices to check that $\liminf_{t\to\infty}X^*_t=1$ whenever $\sup_t X_t>1$.
  We will prove that $X^*_t=1$ whenever $\sup_{s\le t}X_s>1$.
  Fix a $t$.
  The proof is by transfinite induction.
  For nonnegative simple capital processes this is true by definition.
  Now let $X$ be a class $\alpha$ test $I$-supermartingale such that $\sup_{s\le t}X_s>1$.
  Fix $s\le t$ such that $X_s>1$.
  Then $X_s=\liminf_{k\to\infty}X^k_s$ for the fixed $X^k$ of smaller classes,
  and we have $X^k_s>1$ from some $k$ on.
  By the inductive assumption,
  $(X^k)^*_t=1$ from some $k$ on, which implies $X^*_t=1$.
\end{proof}

We call a subset of $[0,\infty)\times\Omega$ a \emph{time-dependent property} of the prices $I$ and $S$.
We say that a time-dependent property $E$ holds \emph{quasi-always} (q.a.)\
if there exists a test supermartingale 
(or equivalently, a test $I$-supermartingale) $X$ such that $X_0=1$ and,
for all $t\in[0,\infty)$ and $\omega\in\Omega$,
\begin{equation}\label{eq:rich}
  \left(
    \exists s<t:(s,\omega)\notin E
  \right)
  \Longrightarrow
  X_t(\omega) = \infty.
\end{equation}
To put it differently, the trader can become infinitely rich as soon as such $E$ is violated.
Many of the results in this paper involve showing that some time-dependent property
(such as the existence of relative quadratic variation or growth rate up to time $t$)
holds quasi-always and therefore can be expected to hold always under the EIH.

\begin{remark}
  In previous work on the topics of this paper (see for example \cite{GTP28}), 
  we used only the very weak form of the EIH that states
  an event will not happen if it allows a trader to become
  infinitely rich infinitely quick.
  This weak hypothesis follows from the efficiency hypothesis for cash
  (which implies that an event with $\UpProb$-probability zero will not happen)
  just as easily as from the EIH for $I$
  (which implies that an event with $\UpProbI$-probability zero will not happen).
  Indeed, if we set
  \[
    \mathrm{Fail}_E:= \{\omega\in\Omega \st (s,\omega)\not\in E \text{ for some } s\in[0,\infty)\}
  \]
  when $E$ is a time-dependent property,
  then $E$ holding quasi-always implies $\UpProb(\mathrm{Fail}_E)=\UpProbI(\mathrm{Fail}_E)=0$;
  see \eqref{eq:rich}.
  For this reason,
  we did not introduce $\UpProbI$ in this previous work.
  Instead we discussed our results
  in terms of $\UpProb$, which is easier to define.
\end{remark}

\section{Existence of some basic quantities (1)}
\label{sec:existence-1}

In this section we do the preparatory work needed to state our results about the equity premium;
namely, we show the existence of all the quantities required in their statements.

All quantities will be defined in terms of the sequences of stopping times $T^n_0:=0$ and
\begin{equation}\label{eq:T-1}
  T^n_k(\omega)
  :=
  \inf
  \Bigl\{
    t>T^n_{k-1}(\omega)
    \st
    \left|I(t)-I(T^n_{k-1})\right| = 2^{-n}
  \Bigr\}
\end{equation}
for $k=1,2,\ldots$;
here $n$ is a positive integer, $n\in\{1,2,\ldots\}$.
The quadratic variation of $I$ on the log scale
(or \emph{relative quadratic variation} of $I$)
can be measured by the sums of squares of the relative increments of $I(t)$,
\begin{equation}\label{eq:Sigma}
  \Sigma^{I,n}_t(\omega)
  :=
  \sum_{k=1}^{\infty}
  \left(
    \frac
    {
      I(T^n_k\wedge t)
      -
      I(T^n_{k-1}\wedge t)
    }
    {
      I(T^n_{k-1}\wedge t)
    }
  \right)^2,
  \quad
  n=1,2,\ldots\,.
\end{equation}
It follows from Theorem~3.1 in \cite{GTP28} and the properties of measure-theoretic Brownian motion
that the limit of $\Sigma^{I,n}$ as $n\to\infty$ exists quasi-always;
however, we will also check this independently in Section~\ref{sec:existence-2}.
The limit will be denoted $\Sigma^I_t(\omega)$.
Moreover, the convergence is uniform on compact intervals,
so the limit is continuous quasi-always.
(Formally, the property ``$\Sigma^{I,n}_s\to\Sigma^I_s$ as $n\to\infty$ uniformly over $s\in[0,t]$''
of $t$ and $\omega$ holds quasi-always.)

\begin{remark}
  Of course, we could have defined $\Sigma^I_t$ as the limit of
  \begin{equation*}
    \sum_{k=1}^{\infty}
    \left(
      \ln I(T^n_k\wedge t)
      -
      \ln I(T^n_{k-1}\wedge t)
    \right)^2
  \end{equation*}
  as $n\to\infty$.
\end{remark}

The quantity $\Sigma^I_t$ measures the accumulated volatility of the index by time $t$.
It can be interpreted as the intrinsic time that elapsed by the moment $t$ of physical time;
unlike physical time, intrinsic time flows faster during intensive trading.

We will simplify our exposition by requiring
that Reality ensure that the function $\Sigma^I$ exists and $\Sigma^I_{\infty}=\infty$.
(Essentially, that the market exists forever and trading in it never dies out.)

The cumulative relative growth of the index $I$ by time $t$ is
\begin{equation}\label{eq:Mu}
  \Mu^{I,n}_t(\omega)
  :=
  \sum_{k=1}^{\infty}
  \frac
  {
    I(T^n_k\wedge t)
    -
    I(T^n_{k-1}\wedge t)
  }
  {
    I(T^n_{k-1}\wedge t)
  },
  \quad
  n=1,2,\ldots
\end{equation}
(where $\Mu$ is the capital version of the Greek letter $\mu$).
The existence of the limit of $\Mu^{I,n}$ as $n\to\infty$ and a simple expression for it
are provided by the following lemma.

\begin{lemma}\label{lem:Mu}
  The limit $\Mu^I:=\lim_{n\to\infty}\Mu^{I,n}$ exists and satisfies, quasi-always,
  \begin{equation*}
    \Mu^I_t
    =
    \ln I(t)
    +
    \frac12
    \Sigma^I_t.
  \end{equation*}
\end{lemma}

\begin{proof} 
  Let us show that the limit exists and is uniform over compact time intervals quasi-always.
  Applying
  \[
    \ln(1+m_k)
    =
    m_k
    -
    \frac12 
    m_k^2
    +
    O\left(\left|m_k\right|^3\right)
  \]
  to
  \begin{equation}\label{eq:m}
    m_k
    :=
    \frac
    {
      I(T^n_k\wedge t)
      -
      I(T^n_{k-1}\wedge t)
    }
    {
      I(T^n_{k-1}\wedge t)
    },
  \end{equation}
  we obtain
  \begin{multline*}
    \ln I(t)
    =
    \sum_{k=1}^{\infty}
    \ln\frac{I(T^n_k\wedge t)}{I(T^n_{k-1}\wedge t)}\\
    =
    \Mu^{I,n}_t
    -
    \frac12
    \Sigma^{I,n}_t
    +
    O
    \left(
      \sum_{k=1}^{\infty}
      \left|
        \frac
        {
          I(T^n_k\wedge t)
          -
          I(T^n_{k-1}\wedge t)
        }
        {
          I(T^n_{k-1}\wedge t)
        }
      \right|^3
    \right),
  \end{multline*}
  and it remains to notice that the last added, $O(\cdots)$,
  is $o(1)$ since the denominator in it can be ignored
  (remember that $I$ is positive).
\end{proof}

\section{Equity premium (1): reduction to the Black--Scholes model}
\label{sec:EP}

In this section we will state two forms of our equity premium result:
as a central limit theorem (which is trivial in the context of Brownian motion)
and as a law of the iterated logarithm.
Remember that we assume that $I(0)=1$.

\begin{lemma}\label{lem:BM}
  Set $\tau_t:=\inf\{s\st \Sigma^I_s\ge t\}$ for each $t\in[0,\infty)$.
  As function of $t$,
  $\ln I(\tau_t) + t/2$ is Brownian motion with respect to $\UpProb$.
\end{lemma}

\begin{remark}
  Formally, Lemma~\ref{lem:BM} says that the pushforward of the upper probability $\UpProb$
  is the standard Wiener measure $\Wiener$ on $C[0,\infty)$ under the following mapping $\phi:\Omega\to C[0,\infty)$:
  $\omega=(I,S)\in\Omega$ is mapped to the path $t\in[0,\infty)\mapsto \ln I(\tau_t) + t/2$.
  The domain of $\phi$ is the set of $(I,S)$ such that $\Sigma^I_{\infty}=\infty$,
  which was our requirement for Reality in Section~\ref{sec:existence-1}.
\end{remark}

\begin{proof}[Proof of Lemma~\ref{lem:BM}]
  Let $E$ be a Borel set in $C[0,\infty)$.
  The set $\phi^{-1}(E)$ is time-superinvariant
  (as defined in \cite{GTP28}, Section~3).
  According to Theorem~3.1 in \cite{GTP28},
  $\UpProb(\phi^{-1}(E))$ coincides with the standard Wiener measure $\Wiener(\phi^{-1}(E))$ of $\phi^{-1}(E)$.
  This measure $\Wiener$ is concentrated on the positive functions $f$
  whose quadratic variation is the identity.
  Applying the time transformation $f'(t'):=f(t)$, where $t':=\int_0^t f^{-2}(s)\dd s$ to those functions,
  we obtain a probability measure $P$ (the pushforward of $\Wiener$ under $f\mapsto f'$)
  concentrated on the functions whose relative quadratic variation $\Sigma$ is the identity;
  we know that $\Wiener(\phi^{-1}(E))=P(\phi^{-1}(E))$.
  By the standard measure-theoretic Dubins--Schwarz theorem,
  $P$ will coincide with the distribution of the measure-theoretic martingale
  \begin{equation}\label{eq:GBM}
    G_t
    :=
    e^{W_t-t/2},
  \end{equation}
  $W$ being the standard Brownian motion (started at $0$).
  (Notice that \eqref{eq:GBM} is a special case of geometric Brownian motion, i.e., the Black--Scholes model.)
  Therefore,
  \[
    \UpProb(\phi^{-1}(E))
    =
    \Wiener(\phi^{-1}(E))
    =
    P(\phi^{-1}(E))
    =
    \Wiener(E).
    \qedhere
  \]
\end{proof}

\begin{remark}\label{rem:BM}
  Lemma~\ref{lem:BM} can be strengthened to say that,
  for any nonnegative time-superinvariant Borel functional $F:\Omega\to[0,\infty)$,
  $
    \UpExpect(F\circ\phi)
    =
    \int F\dd\Wiener
  $.
  Moreover, the same argument as in the proof of Lemma~\ref{lem:BM}
  (but using Theorem~6.3 instead of Theorem~3.1 in \cite{GTP28})
  shows that the last line of the proof can be replaced by
  \[
    \UpExpect(F\circ\phi)
    =
    \int (F\circ\phi) \dd\Wiener
    =
    \int (F\circ\phi) \dd P
    =
    \int F \dd\Wiener.
  \]
\end{remark}

\begin{corollary}\label{cor:BM}
  As function of $t$,
  $\ln I(\tau_t) - t/2$ is Brownian motion with respect to $\UpProbI$.
\end{corollary}

\begin{proof}
  According to Lemma~\ref{lem:BM},
  $W_t:=\ln I(\tau_t)+t/2$ is standard Brownian motion w.r.\ to $\UpProb$.
  To change the num\'eraire we apply Girsanov's theorem
  (see, e.g., \cite{Karatzas/Shreve:1991}, Corollary~3.5.2;
  this version, unlike Theorem~3.5.1 in \cite{Karatzas/Shreve:1991},
  does not require the usual conditions).
  It suffices to show, for each $T>0$,
  that $\ln I(\tau_t) - t/2$, $t\in[0,T]$,
  is Brownian motion over $[0,T]$ with respect to $\UpProbI$.
  Fix such a $T$.
  By Girsanov's theorem, $\tilde W_t:=W_t-t=\ln I(\tau_t)-t/2$ is standard Brownian motion
  w.r.\ to the measure $\tilde P$ on $C[0,T]$
  whose density with respect to $P$
  (as defined in the proof of Lemma~\ref{lem:BM} but restricted to $C[0,T]$) is
  \[
    Z_T
    :=
    e^{W_T-T/2}
    =
    I(\tau_T).
  \]
  Let $\tilde\phi:\Omega\to C[0,T]$ map each $\omega=(I,S)\in\Omega$
  to the path $t\in[0,T]\mapsto \ln I(\tau_t) - t/2$;
  we will continue to use the notation $\phi$ for the function that maps each $\omega=(I,S)\in\Omega$
  to the path $t\in[0,T]\mapsto \ln I(\tau_t) + t/2$.
  Let us check that $\tilde P$ is the pushforward of $\UpProbI$ under $\tilde\phi$:
  for each Borel $E\subseteq C[0,T]$,
  by the definition~\eqref{eq:upper-probability-I} and Remark~\ref{rem:BM},
  \begin{align*}
    \UpProbI(\tilde\phi^{-1}(E))
    &=
    \inf
    \bigl\{
      X_0
      \bigm|
      \forall\omega\in\Omega:
      X_{\tau_T}(\omega) / I_{\tau_T}(\omega)
      \ge
      \III_{\tilde\phi^{-1}(E)}(\omega)
    \bigr\}\\
    &=
    \inf
    \bigl\{
      X_0
      \bigm|
      \forall\omega\in\Omega:
      X_{\tau_T}(\omega)
      \ge
      I_{\tau_T}(\omega)
      \III_{\tilde\phi^{-1}(E)}(\omega)
    \bigr\}\\
    &=
    \UpExpect
    \bigl(
      I_{\tau_T}
      \III_{\tilde\phi^{-1}(E)}
    \bigr)
    =
    \UpExpect
    \bigl(
      I_{\tau_T}
      \III_E \circ {\tilde\phi}
    \bigr)\\
    &=
    \Expect_P
    \bigl(
      Z_T
      \III_E \circ {\tilde W}
    \bigr)
    =
    \Expect_{\tilde P}
    \bigl(
      \III_E \circ {\tilde W}
    \bigr)
    =
    \Wiener(E).
    \qedhere
  \end{align*}
\end{proof}

Let us first derive a central limit theorem for the index from Corollary~\ref{cor:BM}.
Let $z_p$ be the upper $p$-quantile of the standard Gaussian distribution $N_{0,1}$;
i.e., $z_p$ is defined by the requirement that $\Prob(\xi\ge z_{p})=p$,
where $\xi\sim N_{0,1}$.

\begin{corollary}\label{cor:2-sided}
  If $\delta>0$ and $T>0$ are positive constants,
  \begin{equation}\label{eq:2-sided}
    \LowProbI
    \left\{
      \left|\ln I(\tau_T) - T/2\right|
      <
      z_{\delta/2}\sqrt{T}
    \right\}
    =
    1-\delta.
  \end{equation}
\end{corollary}

\noindent
We can interpret~\eqref{eq:2-sided} by saying
that there is a prudent trading strategy that beats the index
by a factor of nearly $1/\delta$ unless
$I(\tau_T)$ is close to $e^{T/2}$ in the sense
\begin{equation*}
  I(\tau_T)
  \in
  \left(
    e^{T/2-z_{\delta/2}\sqrt{T}},
    e^{T/2+z_{\delta/2}\sqrt{T}}
  \right).
\end{equation*}
In other words,
the efficient index can be expected to outperform cash $e^{T/2}$-fold.
The case $\delta\ge1$ in Corollary~\ref{cor:2-sided} is trivial,
but we do not exclude it to simplify its statement;
the upper quantile $z_p$ is understood to be $-\infty$ when $p\ge1$.

If we are only interested in a lower or upper bound on $I$,
we can use the following corollary.
\begin{corollary}\label{cor:1-sided}
  Let $\delta$ and $T$ be as before.
  Then
  \begin{equation*}
    \LowProbI
    \left\{
      \ln I(\tau_T) - T/2
      >
      -z_{\delta}\sqrt{T}
    \right\}
    =
    1-\delta
  \end{equation*}
  and
  \begin{equation*}
    \LowProbI
    \left\{
      \ln I(\tau_T) - T/2
      <
      z_{\delta}\sqrt{T}
    \right\}
    =
    1-\delta.
  \end{equation*}
\end{corollary}

Corollaries~\ref{cor:2-sided} and~\ref{cor:1-sided} follow immediately from Corollary~\ref{cor:BM}.
Corollary~\ref{cor:BM} also immediately implies
the following law of the iterated logarithm for the equity premium:

\begin{corollary}\label{cor:EP-LIL}
  It is $\UpProbI$-almost certain that
  \[
    \limsup_{t\to\infty}
    \frac{\ln I(t) - \Sigma^I_t/2}{\sqrt{2 \Sigma^I_t \ln\ln \Sigma^I_t}}
    =
    1
  \]
  and
  \[
    \liminf_{t\to\infty}
    \frac{\ln I(t) - \Sigma^I_t/2}{\sqrt{2 \Sigma^I_t \ln\ln \Sigma^I_t}}
    =
    -1.
  \]
\end{corollary}

\section{Equity-premium (2): mixing method}
\label{sec:EPM}

In this section we will discuss an alternative approach to the equity premium phenomenon,
which will also be used to derive a probability-free CAPM
in Section~\ref{sec:CAPM}.
The test $I$-supermartingale whose existence is implicitly asserted in~\eqref{eq:2-sided}
is, in a sense, reckless:
it beats the index $1/\delta$-fold if the event in the curly braces fails to happen
but can (and does) lose everything if it happens.
In this section we will discuss safer (more conservative) trading strategies
instead of ``all-or-nothing'' trading strategies
fine-tuned to the event of interest (such as the one in~\eqref{eq:2-sided}).

\begin{lemma}\label{lem:supermartingale-1}
  For each $\epsilon\in\bbbr$, the process
  \begin{equation}\label{eq:supermartingale-1}
    \exp
    \left(
      \epsilon (\Mu^I_t - \Sigma^I_t)
      -
      \frac{\epsilon^2}{2} \Sigma^I_t
    \right)
  \end{equation}
  is a test $I$-supermartingale q.a.
\end{lemma}

\noindent
In other words,
the lemma says that the process~\eqref{eq:supermartingale-1}
coincides with a test $I$-supermartingale quasi-always.

\begin{proof}
  Let us show that \eqref{eq:supermartingale-1} is a class 1 test $I$-supermartingale
  (see Remark~\ref{rem:I-supermartingales} for the definition).

  The value of the index $I$ at time $T^n_K$ is $\prod_{k=1}^K (1+m_k)$
  (where $m_k$ is defined by~\eqref{eq:m}).
  Let us consider the simple capital process whose value at time $T^n_K$ is
  $\prod_{k=1}^K (1+(1+\epsilon)m_k)$
  (which should be stopped as soon as the capital hits 0).
  Intuitively, we are mixing the returns $m_k$ of $I$ and the returns $0$ of cash
  (and this is a convex mixture when $\epsilon\in[-1,0]$);
  when $\left|\epsilon\right|$ is small
  (which case is important for limit theorems such as the law of the iterated logarithm
  in Corollary~\ref{cor:EP-LIL-weak}),
  the new simple capital process can be regarded as a perturbation of $I$.

  We can see that
  \[
    \ln \prod_{k=1}^K (1+(1+\epsilon)m_k)
    -
    \ln \prod_{k=1}^K (1+m_k)
  \]
  is the value at time $T^n_K$ of the log of a test $I$-supermartingale (of class 0).
  In combination with Taylor's expansion, this implies that
  \[
    \epsilon \sum_{k=1}^K m_k
    -
    \epsilon \sum_{k=1}^K m_k^2
    -
    \frac{\epsilon^2}{2} \sum_{k=1}^K m_k^2
    +
    O
    \left(
      \sum_{k=1}^K \left|m_k\right|^3
    \right)
  \]
  is the value at time $T^n_K$ of the log of a test $I$-supermartingale.
  Passing to the limit as $n\to\infty$ (and remembering that the variation index of $I$
  over compact intervals does not exceed 2 quasi-always \cite{GTP28}),
  we obtain that
  \[
    \epsilon (\Mu^I_t - \Sigma^I_t)
    -
    \frac{\epsilon^2}{2} \Sigma^I_t
  \]
  is the log of a test $I$-supermartingale (of class 1) q.a.
\end{proof}

\begin{corollary}
  For any $\epsilon>0$ and $\delta>0$,
  \[
    \LowProbI
    \left\{
      \forall t\in[0,\infty):
      \left|
        \Mu^I_t - \Sigma^I_t
      \right|
      <
      \frac{1}{\epsilon}\ln\frac{2}{\delta}
      +
      \frac{\epsilon}{2} \Sigma^I_t
    \right\}
    \ge
    1-\delta.
  \]
\end{corollary}

\begin{proof}
  Fix $\epsilon>0$ and $\delta>0$.
  By Lemmas~\ref{lem:robustness} and~\ref{lem:supermartingale-1},
  with lower $I$-probability at least $1-\delta/2$ we will have
  \[
    \forall t\in[0,\infty):
    \epsilon (\Mu^I_t - \Sigma^I_t)
    -
    \frac{\epsilon^2}{2} \Sigma^I_t
    <
    \ln\frac{2}{\delta}.
  \]
  Dividing both sides by $\epsilon$
  and considering the same test $I$-supermartingale but with $-\epsilon$ in place of $\epsilon$,
  we obtain
  \[
    \forall t\in[0,\infty):
    \left|
      \Mu^I_t - \Sigma^I_t
    \right|
    <
    \frac{1}{\epsilon}\ln\frac{2}{\delta}
    +
    \frac{\epsilon}{2} \Sigma^I_t,
  \]
  with lower probability at least $1-\delta$.
\end{proof}

The following corollary is in the spirit of Corollary~\ref{cor:2-sided}.

\begin{corollary}\label{cor:EP-given-epsilon}
  If $\delta>0$, $\epsilon>0$, and $\tau_T:=\inf\{t\st \Sigma^I_t\ge T\}$ for some constant $T>0$,
  \begin{equation}\label{eq:EP-given-epsilon}
    \LowProbI
    \left\{
      \left|
        \Mu^I_{\tau_T} - T
      \right|
      <
      \frac{1}{\epsilon}\ln\frac{2}{\delta}
      +
      \frac{\epsilon}{2} T
    \right\}
    \ge
    1-\delta.
  \end{equation}
\end{corollary}

It is natural to optimize the $\epsilon$ in \eqref{eq:EP-given-epsilon} given $\delta$ and $T$,
\begin{equation}\label{eq:min}
  \min_{\epsilon>0}
  \left(
    \frac{1}{\epsilon}\ln\frac{2}{\delta}
    +
    \frac{\epsilon}{2} T
  \right)
  =
  \sqrt{2T\ln\frac{2}{\delta}},
\end{equation}
which gives our final corollary in this direction.

\begin{corollary}
  If $\delta>0$ and $\tau_T:=\inf\{t\st \Sigma^I_t\ge T\}$ for some constant $T>0$,
  \begin{equation*}
    \LowProbI
    \left\{
      \left|
        \Mu^I_{\tau_T} - T
      \right|
      <
      \sqrt{2T\ln\frac{2}{\delta}}
    \right\}
    \ge
    1-\delta,
  \end{equation*}
  or, equivalently (by Lemma~\ref{lem:Mu}),
  \begin{equation}\label{eq:EP-tuned-epsilon}
    \LowProbI
    \left\{
      \left|
        \ln I_{\tau_T} - T/2
      \right|
      <
      \sqrt{2T\ln\frac{2}{\delta}}
    \right\}
    \ge
    1-\delta.
  \end{equation}
\end{corollary}

It is instructive to compare~\eqref{eq:EP-tuned-epsilon} and \eqref{eq:2-sided},
which are obtained using two very different methods
(especially that for the more general CAPM-type results of the following sections
we will be able to use only the second, more conservative, method).
The difference between the two inequalities for $I(\tau_T)$
asserted with lower probability of (at least) $1-\delta$
boils down to the difference between the functions
\[
  X(q):=z_q \text{ and } \eta(q):=\sqrt{2\ln\frac{1}{q}},
\]
in the notation of Hastings \cite{Hastings:1955},
who considers $q\in(0,0.5]$.
Hastings gives two approximations to $X(q)$
(pp.~191--192, reproduced in \cite{Abramowitz/Stegun:1964}, 26.2.22 and 26.2.23)
as the optimal, in a minimax sense,
product of $\eta(q)$ and a rational function of $\eta(q)$
(the two approximations correspond to different degrees of the polynomials in the numerator and denominator
of the rational function).

It is easy to check that $X(q)\sim\eta(q)$ as $q\to0$.
Figure~\ref{fig:inequality} compares the two functions over the range $q\in[10^{-5},0.5]$
($q=0.5$ corresponding to the trivial value $\delta=1$).

\begin{figure}[t]
  \begin{center}
    \includegraphics[width=0.7\columnwidth]{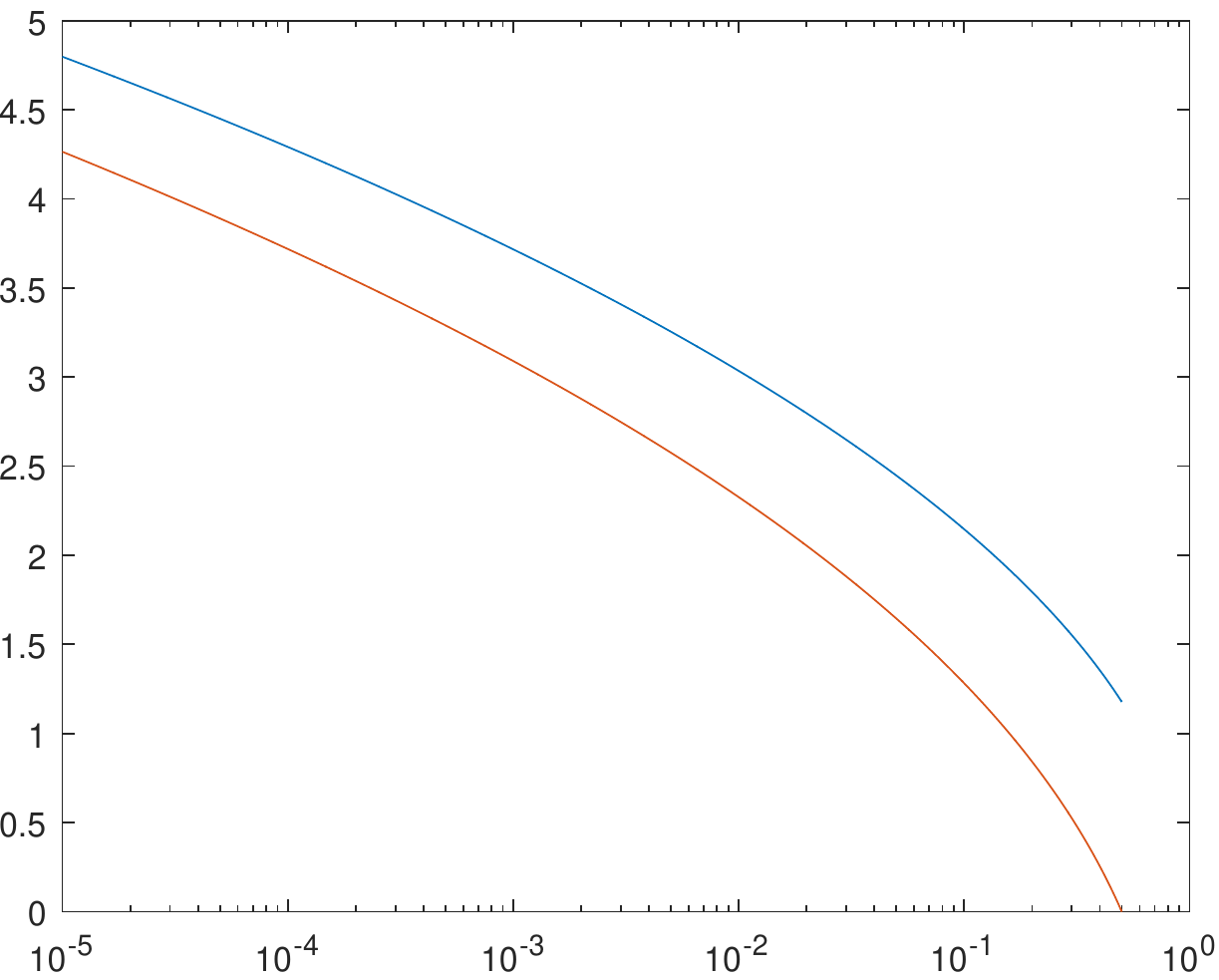}
  \end{center}
  \caption{The functions $X$ (in red, lower) and $\eta$ (in blue, higher) over $[10^{-5},0.5]$}
  \label{fig:inequality}
\end{figure}

It is easy to prove the validity part of a LIL for the equity premium
using Ville's \cite{Ville:1939} method.

\begin{corollary}\label{cor:EP-LIL-weak}
  Almost surely w.r.\ to $\UpProbI$,
  \[
    \limsup_{t\to\infty}
    \frac{\left|\Mu^I_t-\Sigma^I_t\right|}{\sqrt{2\Sigma^I_t\ln\ln\Sigma^I_t}}
    \le
    1.
  \]
\end{corollary}

\begin{proof}
  This is part of Corollary~\ref{cor:EP-LIL}
  (combined with Lemma~\ref{lem:Mu}),
  so there is nothing to prove.
  But alternatively, we could mix the processes
  (which quasi-always coincide with test $I$-supermartingales:
  cf.\ Lemma~\ref{lem:supermartingale-1})
  \begin{equation*}
    \exp
    \left(
      -\epsilon \Mu^I_t
      +\epsilon \Sigma^I_t
      - \frac{\epsilon^2}{2} \Sigma^I_t
    \right)
  \end{equation*}
  and
  \begin{equation*}
    \exp
    \left(
      \epsilon \Mu^I_t
      -\epsilon \Sigma^I_t
      - \frac{\epsilon^2}{2} \Sigma^I_t
    \right)
  \end{equation*}
  over $\epsilon>0$ of the form $(1+\kappa)^{-k}$, $k=1,2,\ldots$,
  with weights $w_k=k^{-1-\delta}$ (so that $w_k\to0$ slowly while still $\sum_kw_k<\infty$)
  for small $\kappa>0$ and $\delta>0$.
  For details, see Section~\ref{sec:CAPM},
  where we will prove a more general statement.
\end{proof}

Intuitively, our new method still allows us to derive the upper LIL
since $X(\delta/2)\sim\eta(\delta/2)$ as $\delta\to0$,
and the LIL is about almost sure convergence and so corresponds to small $\delta$.

\section{Existence of some basic quantities (2)}
\label{sec:existence-2}

This section continues the series of definitions that we started in Section~\ref{sec:existence-1}.
Now we consider another traded security $S$.
Since it is a traded security, we can define $\Sigma^S$ and $\Mu^S$
analogously to $\Sigma^I$ and $\Mu^I$;
however, this would involve stopping times defined in terms of $S$ rather than $I$.
It is more convenient to have just one family of stopping times.
Therefore, we now define the sequence of stopping times $T^n_k$, $k=0,1,2,\ldots$, inductively
by $T^n_{0}(\omega):=0$ and
\begin{equation}\label{eq:T-2}
  T^n_k(\omega)
  :=
  \inf
  \Bigl\{
    t>T^n_{k-1}(\omega)
    \st
    \left|I(t)-I(T^n_{k-1})\right| \vee \left|S(t)-S(T^n_{k-1})\right| = 2^{-n}
  \Bigr\}
\end{equation}
for $k=1,2,\ldots$\,.
We let $T^n(\omega)$ stand for the \emph{$n$th partition}, i.e., the set
\[
  T^n(\omega)
  :=
  \left\{
    T^n_k(\omega)
    \st
    k=0,1,\ldots
  \right\};
\]
under our new definition, the partitions are not necessarily nested,
$T^1\subseteq T^2\subseteq\cdots$
(as was the case for our old definition~\eqref{eq:T-1}).

The following lemma says that we can redefine $\Sigma^{I,n}_t$ as \eqref{eq:Sigma} using the new partitions.

\begin{lemma}\label{lem:Sigma-2}
  The limit in~\eqref{eq:Sigma} exists and is uniform on compact intervals quasi-always
  for both partitions~\eqref{eq:T-1} and~\eqref{eq:T-2};
  the limits coincide quasi-always.
\end{lemma}

\begin{proof}
  As shown in \cite{GTP42}, the limit of \eqref{eq:Sigma} as $n\to\infty$
  exists and is uniform on compact intervals quasi-always if we ignore the denominator;
  namely, the sequence
  \begin{equation*}
    \Alpha^{I,n}_t(\omega)
    :=
    \sum_{k=1}^{\infty}
    \left(
      I(T^n_k\wedge t)
      -
      I(T^n_{k-1}\wedge t)
    \right)^2,
    \quad
    n=1,2,\ldots,
  \end{equation*}
  converges to a function $\Alpha^I$ uniformly on compact intervals quasi-always.
  The function $\Alpha^I$ is the same (quasi-always) for the sequences of partitions~\eqref{eq:T-1} and~\eqref{eq:T-2}:
  to check this, apply the argument given in Section~5 of \cite{GTP42} to the $n$th partitions
  in~\eqref{eq:T-1} and~\eqref{eq:T-2} rather than to the $(n-1)$th and $n$th partitions in the same sequence.
  It is clear that $\Sigma^I$ is the Riemann--Stiltjes integral
  \[
    \Sigma^I_t
    =
    \int_0^t
    \frac{\dd\Alpha^I_s}{I^2(s)}
  \]
  and that the statement of uniform convergence carries over to $\Sigma^{I,n}$.
\end{proof}

Next we state the analogue of Lemma~\ref{lem:Sigma-2} for $\Mu^I$.

\begin{lemma}\label{lem:Mu-2}
  The limit in~\eqref{eq:Mu} exists and is uniform on compact intervals quasi-always
  for both partitions~\eqref{eq:T-1} and~\eqref{eq:T-2};
  the limits coincide quasi-always.
\end{lemma}

\begin{proof}
  The existence of the limit quasi-always is shown in \cite{GTP42}, Section~4
  (and the earlier work \cite{Perkowski/Promel:2016} by Perkowski and Pr\"omel);
  the limit is nothing else than the It\^o integral
  \[
    \Mu^I_t
    =
    \int_0^t
    \frac{\dd I(s)}{I(s)}.
  \]
  The coincidence of the functions $\Mu^I$ quasi-always
  for the sequences of partitions~\eqref{eq:T-1} and~\eqref{eq:T-2}
  follows from the argument given in Section~4 of \cite{GTP42}
  applied to the $n$th partitions in~\eqref{eq:T-1} and~\eqref{eq:T-2}
  rather than to the $(n-1)$th and $n$th partitions in the same sequence.
\end{proof}

Using $S$ in place of $I$,
we obtain the definitions of $\Sigma^S$ and $\Mu^S$.
The analogues of Lemmas~\ref{lem:Sigma-2} and~\ref{lem:Mu-2} still hold.

As we are also interested in the covariance between (the returns of) $S$ and $I$, we define
\begin{multline}\label{eq:Sigma-ST}
  \Sigma^{S,I,n}_t(\omega)
  :=
  \sum_{k=1}^{\infty}
  \frac
  {
    S(T^n_k\wedge t)
    -
    S(T^n_{k-1}\wedge t)
  }
  {
    S(T^n_{k-1}\wedge t)
  }
  \frac
  {
    I(T^n_k\wedge t)
    -
    I(T^n_{k-1}\wedge t)
  }
  {
    I(T^n_{k-1}\wedge t)
  },\\
  n=1,2,\ldots,
\end{multline}
and then set $\Sigma^{S,I}$ to the limit as $n\to\infty$.
The existence of the limit quasi-always is asserted in our next lemma.

\begin{lemma}
  The limit of~\eqref{eq:Sigma-ST} as $n\to\infty$ exists quasi-always
  uniformly on compact intervals.
\end{lemma}

\begin{proof}
  First we notice that, if the denominators in~\eqref{eq:Sigma-ST} are ignored,
  the limit $\Alpha^{S,I}_t(\omega)$ of
  \begin{multline*}
    \Alpha^{S,I,n}_t(\omega)
    :=
    \sum_{k=1}^{\infty}
    \left(
      S(T^n_k\wedge t)
      -
      S(T^n_{k-1}\wedge t)
    \right)
    \left(
      I(T^n_k\wedge t)
      -
      I(T^n_{k-1}\wedge t)
    \right)\\
    =
    \frac12
    \left(
      \Alpha^{S+I,n}_t
      -
      \Alpha^{S,n}_t
      -
      \Alpha^{I,n}_t
    \right)
  \end{multline*}
  will exist quasi-always uniformly on compact intervals;
  moreover, it will have bounded variation on compact intervals q.a.
  It remains to notice that \eqref{eq:Sigma-ST} are approximating sums for the Riemann--Stiltjes integral
  \[
    \Sigma^{S,I}_t
    =
    \int_0^t
    \frac{\dd\Alpha^{S,I}_s}{S(s)I(s)}.
    \qedhere
  \]
\end{proof}

The last quantity that we will need quantifies the difference between the returns of $S$ and $I$:
\begin{multline}\label{eq:Delta-ST}
  \Delta^{S,I,n}_t(\omega)
  :=
  \sum_{k=1}^{\infty}
  \left(
    \frac
    {
      S(T^n_k\wedge t)
      -
      S(T^n_{k-1}\wedge t)
    }
    {
      S(T^n_{k-1}\wedge t)
    }
    -
    \frac
    {
      I(T^n_k\wedge t)
      -
      I(T^n_{k-1}\wedge t)
    }
    {
      I(T^n_{k-1}\wedge t)
    }
  \right)^2,\\
  n=1,2,\ldots;
\end{multline}
we then set $\Delta^{S,I}_t$ to the limit as $n\to\infty$.
The final lemma of this section shows that the limit as $n\to\infty$ exists quasi-always and is closely related
to the other quantities introduced in this section.

\begin{lemma}\label{lem:VSI}
  The limit $\Delta^{S,I}_t$ of~\eqref{eq:Delta-ST} as $n\to\infty$
  exists uniformly on compact intervals quasi-always
  and satisfies
  \[
    \Delta^{S,I}_t
    =
    \Sigma^{S}_t
    +
    \Sigma^{I}_t
    -
    2 \Sigma^{S,I}_t
    \quad
    \text{q.a.}
  \]
\end{lemma}

\begin{proof}
  It suffices to notice that, using the notation $m_k$ and $s_k$ for the returns of $S$ and $I$
  ($m_k$ is defined by \eqref{eq:m} and $s_k$ is defined in the same way using $S$ in place of $I$),
  \[
    \Delta^{S,I,n}_t
    =
    \sum_{k=1}^{\infty}
    (s_k-m_k)^2
    =
    \sum_{k=1}^{\infty}
    s_k^2
    +
    \sum_{k=1}^{\infty}
    m_k^2
    -
    2
    \sum_{k=1}^{\infty}
    s_k m_k,
  \]
  and pass to the limit as $n\to\infty$.
\end{proof}

\section{Capital Asset Pricing Model}
\label{sec:CAPM}

Let us use the same mixing method as in Section~\ref{sec:EPM},
except that now we will apply it to $I$ and $S$ rather than to $I$ and cash.

\begin{lemma}\label{lem:supermartingale-2}
  For each $\epsilon\in\bbbr$, the process
  \begin{equation}\label{eq:supermartingale-2}
    \exp
    \left(
      \epsilon (\Mu^S_t - \Mu^I_t + \Sigma^I_t - \Sigma^{S,I}_t)
      -
      \frac{\epsilon^2}{2} \Delta^{S,I}_t
    \right)
  \end{equation}
  is a test $I$-supermartingale q.a.
\end{lemma}

\begin{proof}
  The value of the index $I$ at time $T^n_K$ is
  $\prod_{k=1}^K (1+m_k)$,
  and the value of the security $S$ at time $T^n_K$ is
  $\prod_{k=1}^K (1+s_k)$,
  where as before we use $s_k$ for the analogue of $m_k$ for $S$.
  Let us consider the simple capital process whose value at time $T^n_K$ is
  $\prod_{k=1}^K (1+(1-\epsilon)m_k+\epsilon s_k)$
  (except that it is stopped if and when it hits $0$);
  it can be considered to be a mixture (convex mixture when $\epsilon\in[0,1]$) of $I$ and $S$.
  We can see that
  \[
    \ln \prod_{k=1}^K (1+(1-\epsilon)m_k+\epsilon s_k)
    -
    \ln \prod_{k=1}^K (1+m_k)
  \]
  is the log of a test $I$-supermartingale at time $T^n_K$ for all $K$,
  which implies the analogous statement for
  \begin{multline*}
    \epsilon \sum_{k=1}^K s_k
    -
    \epsilon \sum_{k=1}^K m_k
    +
    \epsilon \sum_{k=1}^K m_k^2
    -
    \epsilon \sum_{k=1}^K s_k m_k\\
    -
    \frac{\epsilon^2}{2} \sum_{k=1}^K m_k^2
    +
    \epsilon^2 \sum_{k=1}^K s_k m_k
    -
    \frac{\epsilon^2}{2} \sum_{k=1}^K s_k^2\\
    +
    O
    \left(
      \sum_{k=1}^K \left|s_k\right|^3
    \right)
    +
    O
    \left(
      \sum_{k=1}^K \left|m_k\right|^3
    \right).
  \end{multline*}
  Passing to the limit as $n\to\infty$,
  we obtain that
  \[
    \epsilon \Mu^S_t
    -
    \epsilon \Mu^I_t
    +
    \epsilon \Sigma^I_t
    -
    \epsilon \Sigma^{S,I}_t
    -
    \frac{\epsilon^2}{2} \Delta^{S,I}_t
  \]
  is the log of a test $I$-supermartingale q.a.
\end{proof}

\begin{corollary}\label{cor:1}
  For any $\epsilon>0$ and $\delta>0$,
  \[
    \LowProbI
    \left\{
      \forall t\in[0,\infty):
      \left|
        \Mu^S_t - \Mu^I_t + \Sigma^I_t - \Sigma^{S,I}_t
      \right|
      <
      \frac{1}{\epsilon}\ln\frac{2}{\delta}
      +
      \frac{\epsilon}{2} \Delta^{S,I}_t
    \right\}
    \ge
    1-\delta.
  \]
\end{corollary}

\begin{proof}
  For $\epsilon>0$, the fact that \eqref{eq:supermartingale-2} is a test $I$-supermartingale
  implies that
  \begin{equation}\label{eq:inequality}
    \forall t\in[0,\infty):
    \epsilon
    \left(
      \Mu^S_t - \Mu^I_t + \Sigma^I_t - \Sigma^{S,I}_t
    \right)
    <
    \ln\frac{2}{\delta}
    +
    \frac{\epsilon^2}{2} \Delta^{S,I}_t
  \end{equation}
  with lower $I$-probability at least $1-\delta/2$.
  It remains to divide both sides of the inequality in~\eqref{eq:inequality} by $\epsilon$
  and consider the same test $I$-supermartingale but with $-\epsilon$ in place of $\epsilon$
  (which should be stopped as soon as the capital hits 0).
\end{proof}

The following corollary is in the spirit of Corollaries~\ref{cor:2-sided} and~\ref{cor:EP-given-epsilon};
however, now we wait until $I$ and $S$ become sufficiently different.

\begin{corollary}\label{cor:2}
  If $\delta>0$, $\epsilon>0$, and $\tau_T:=\inf\{t\st \Delta^{S,I}_t\ge T\}$ for some constant $T>0$,
  \begin{equation}\label{eq:CAPM-given-epsilon}
    \LowProbI
    \left\{
      \left|
        \Mu^S_{\tau_T} - \Mu^I_{\tau_T} + \Sigma^I_{\tau_T} - \Sigma^{S,I}_{\tau_T}
      \right|
      <
      \frac{1}{\epsilon}\ln\frac{2}{\delta}
      +
      \frac{\epsilon}{2} T
    \right\}
    \ge
    1-\delta.
  \end{equation}
  Our convention is that the event in the curly braces in~\eqref{eq:CAPM-given-epsilon}
  happens when $\tau_T=\infty$
  (i.e., when the security essentially coincides with the index).
\end{corollary}

It is natural to optimize the $\epsilon$ in \eqref{eq:CAPM-given-epsilon} given $\delta$ and $T$,
as we did in \eqref{eq:min},
which gives us the following corollary.

\begin{corollary}
  If $\delta>0$ and $\tau_T:=\inf\{t\st \Delta^I_t\ge T\}$ for some constant $T>0$,
  \begin{equation}\label{eq:CAPM-optimized}
    \LowProbI
    \left\{
      \left|
        \Mu^S_{\tau_T} - \Mu^I_{\tau_T} + \Sigma^I_{\tau_T} - \Sigma^{S,I}_{\tau_T}
      \right|
      <
      \sqrt{2T\ln\frac{2}{\delta}}
    \right\}
    \ge
    1-\delta.
  \end{equation}
\end{corollary}

We will now give a complete proof of the validity part of a LIL for the CAPM;
the following proposition generalizes Corollary~\ref{cor:EP-LIL-weak}
(we obtain the latter by taking cash as $S$).

\begin{proposition}\label{prop:CAPM-LIL}
  Almost surely w.r.\ to $\UpProbI$,
  \begin{equation}\label{eq:CAPM-LIL}
    \Delta^{S,I}_t\to\infty
    \Longrightarrow
    \limsup_{t\to\infty}
    \frac{\left|\Mu^S_t-\Mu^I_t+\Sigma^I_t-\Sigma^{S,I}_t\right|}{\sqrt{2\Delta^{S,I}_t\ln\ln \Delta^{S,I}_t}}
    \le
    1.
  \end{equation}
\end{proposition}

\begin{proof}
  We will implement in detail the idea mentioned in the proof of Corollary~\ref{cor:EP-LIL-weak},
  namely, we will mix the processes,
  which quasi-always are test $I$-supermartingales
  (see Lemma~\ref{lem:supermartingale-2}),
  \begin{equation}\label{eq:martingale-1}
    \exp
    \left(
      \epsilon \Mu^S_t
      -\epsilon \Mu^I_t
      +\epsilon \Sigma^I_t
      -\epsilon \Sigma^{S,I}_t
      - \frac{\epsilon^2}{2} \Delta^{S,I}_t
    \right)
  \end{equation}
  and
  \begin{equation*}
    \exp
    \left(
      -\epsilon \Mu^S_t
      +\epsilon \Mu^I_t
      -\epsilon \Sigma^I_t
      +\epsilon \Sigma^{S,I}_t
      - \frac{\epsilon^2}{2} \Delta^{S,I}_t
    \right)
  \end{equation*}
  over $\epsilon>0$ of the form $(1+\kappa)^{-k}$, $k=1,2,\ldots$,
  with weights $w_k=k^{-1-\delta}$.
  We will only prove~\eqref{eq:CAPM-LIL} with the operation of taking the absolute value omitted
  (the proof for the case where it is replaced by negation is analogous),
  and so we will only consider~\eqref{eq:martingale-1}.

  Since $\sum_k w_k$ converges,
  \begin{equation*}
    \sum_k
    w_k
    \exp
    \left(
      \epsilon_k \Mu^S_t
      -\epsilon_k \Mu^I_t
      +\epsilon_k \Sigma^I_t
      -\epsilon_k \Sigma^{S,I}_t
      - \frac{\epsilon_k^2}{2} \Delta^{S,I}_t
    \right)
  \end{equation*}
  is also a test $I$-supermartingale with finite initial capital;
  therefore, it is bounded $\UpProbI$-a.s.,
  and so we have
  \begin{equation*}
    \Mu^S_t-\Mu^I_t+\Sigma^I_t-\Sigma^{S,I}_t
    \le
    \frac{-\ln w_k+O(1)}{\epsilon_k}
    +
    \frac{\epsilon_k}{2} \Delta^{S,I}_t
    \quad
    \text{$\UpProbI$-a.s.},
  \end{equation*}
  i.e.,
  \begin{equation*}
    \Mu^S_t-\Mu^I_t+\Sigma^I_t-\Sigma^{S,I}_t
    \le
    \frac{(1+\delta)\ln k+O(1)}{\epsilon_k}
    +
    \frac{\epsilon_k}{2} \Delta^{S,I}_t
    \quad
    \text{$\UpProbI$-a.s.}
  \end{equation*}
  The value of $\epsilon=\epsilon_k$ that minimizes the right-hand side
  (let's forget for a minute that $\epsilon$ is a function of $k$ and ignore the $O(1)$)
  is
  \[
    \epsilon
    =
    \sqrt{2\frac{(1+\delta)\ln k}{\Delta^{S,I}_t}}.
  \]
  Let us choose $k$ such that this is approximately true, namely,
  \begin{equation}\label{eq:approximate}
    \epsilon_{k+2}
    =
    (1+\kappa)^{-k-2}
    \le
    \sqrt{2\frac{(1+\delta)\ln k}{\Delta^{S,I}_t}}
    \le
    (1+\kappa)^{-k}
    =
    \epsilon_{k};
  \end{equation}
  it is easy to check that such $k$ exist
  when $\Delta^{S,I}_t$ is sufficiently large.
  This gives us
  \begin{equation}\label{eq:almost-there}
    \Mu^S_t-\Mu^I_t+\Sigma^I_t-\Sigma^{S,I}_t
    \le
    2(1+\kappa)^2
    \sqrt{(1+\delta)(\ln k+O(1))\frac12 \Delta^{S,I}_t}
    \quad
    \text{$\UpProbI$-a.s.}
  \end{equation}
  The right-hand inequality in~\eqref{eq:approximate} can be rewritten as
  \begin{equation*}
    k\ln(1+\kappa)
    \le
    \frac12
    \left(
      \ln \Delta^{S,I}_t - \ln 2 - \ln(1+\delta) - \ln\ln k
    \right)
    \le
    \frac12
    \ln \Delta^{S,I}_t
  \end{equation*}
  (for large $k$),
  and plugging this into~\eqref{eq:almost-there} gives
  \begin{multline*}
    \Mu^S_t-\Mu^I_t+\Sigma^I_t-\Sigma^{S,I}_t\\
    \le
    2(1+\kappa)^2
    \sqrt{(1+\delta)\left(\ln\frac12+\ln\ln \Delta^{S,I}_t-\ln\ln(1+\kappa)+O(1)\right)\frac12 \Delta^{S,I}_t}\\
    \text{$\UpProbI$-a.s.}
  \end{multline*}
  It remains to mix over sequences of $\kappa\to0$ and $\delta\to0$.
\end{proof}

\subsection*{Theoretical performance deficit}

Using Lemma~\ref{lem:Mu} we can rewrite~\eqref{eq:CAPM-optimized} as
\begin{equation*}
  \LowProbI
  \left\{
    \left|
      \ln S_{\tau_T} - \ln I_{\tau_T} + \frac12 \Sigma^S_{\tau_T} + \frac12 \Sigma^I_{\tau_T} - \Sigma^{S,I}_{\tau_T}
    \right|
    <
    \sqrt{2T\ln\frac{2}{\delta}}
  \right\}
  \ge
  1-\delta.
\end{equation*}
And using Lemma~\ref{lem:VSI} we can rewrite the last inequality as
\begin{equation*}
  \LowProbI
  \left\{
    \left|
      \ln S_{\tau_T} - \ln I_{\tau_T} + \frac12 \Delta^{S,I}_{\tau_T}
    \right|
    <
    \sqrt{2T\ln\frac{2}{\delta}}
  \right\}
  \ge
  1-\delta.
\end{equation*}
Therefore, for large $T$ the EIH for $I$ implies
\[
  \ln S(\tau_T) \approx \ln I(\tau_T) - \frac12 \Delta^{S,I}_{\tau_T}
\]
(the $\approx$ assumes that $\ln S(\tau_T)$, $\ln I(\tau_T)$, and $\Delta^{S,I}_{\tau_T}$
grow linearly in $T$,
and so $\sqrt{T}$ can be regarded as small).
The subtrahend $\frac12 \Delta^{S,I}_{\tau_T}$ can be interpreted
as measuring the lack of diversification in $S$ as compared with $I$;
we call it the \emph{theoretical performance deficit}.

Lemma~\ref{lem:supermartingale-2}, Corollaries~\ref{cor:1}--\ref{cor:2}, and Proposition~\ref{prop:CAPM-LIL}
can also be stated in terms of the theoretical performance deficit.

\subsection*{Acknowledgments}

This research was supported by the Air Force Office of Scientific Research
(grant FA9550-14-1-0043).

\newcommand{\noopsort}[1]{}

\end{document}

%% file: gtarXiv.txt
\makeatletter

\newif\iftwodates
\twodatesfalse

\renewcommand\maketitle{\begin{titlepage}%
  \pagenumbering{Alph}  
  \let\footnotesize\small
  \let\footnoterule\relax
  \let\footnote\thanks
  \null\vfil
  \vskip 30\p@
  \begin{center}%
    {\LARGE \bf \@title \par}%
    \vskip 3em%
    {\large
     \lineskip .75em%
     \begin{tabular}[t]{c}%
       \@author
     \end{tabular}\par}%
     \vskip 1.5em%
  \end{center}\par
  \vfill
  \begin{center}
    \raisebox{1.5cm}{\includegraphics[width=0.58\textwidth]%
      {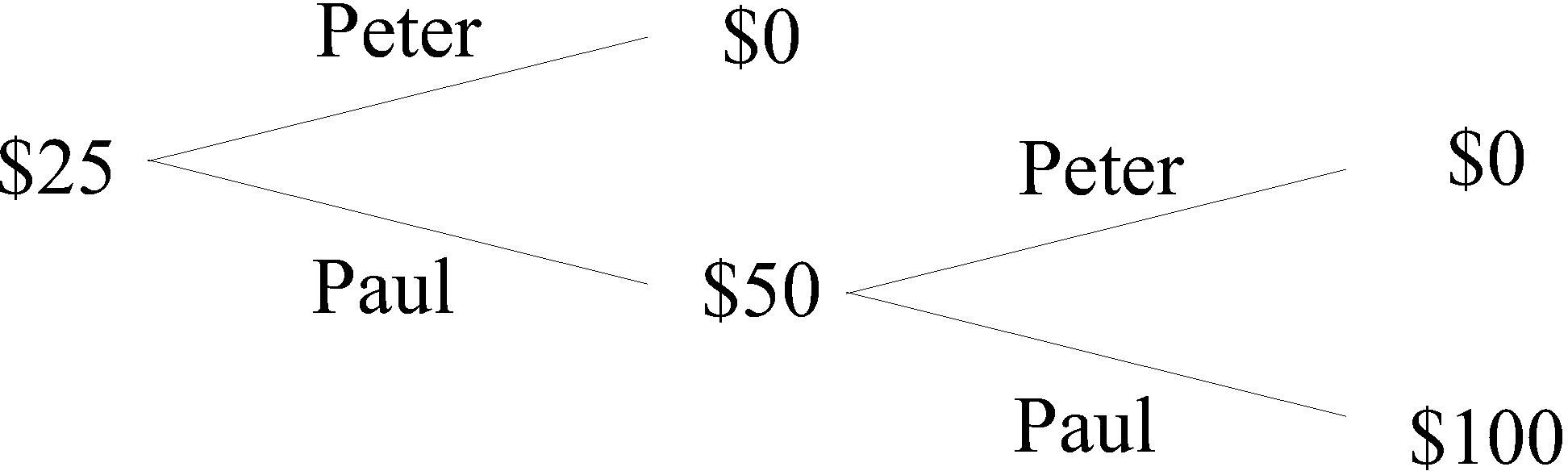}}%
    \hskip 3em%
    \includegraphics[width=0.29\textwidth]%
      {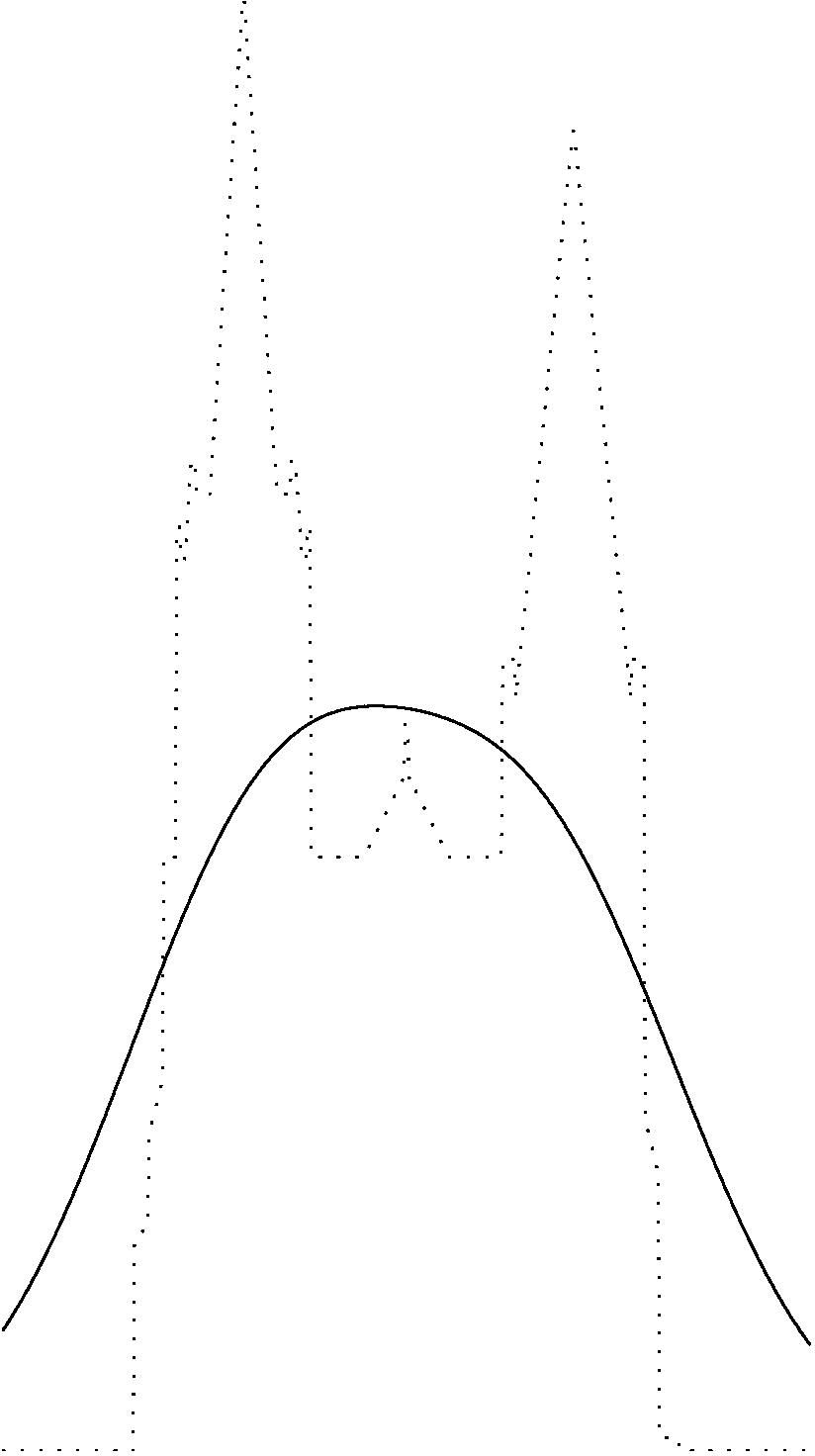}%
  \end{center}
  \@thanks
  \vfill
  \begin{center}
    {\large \bf The Game-Theoretic Probability and Finance Project}
  \end{center}
  \begin{center}
    {\large Working Paper \#\No}
  \end{center}
  \begin{center}
    {\iftwodates\large First posted \firstposted.
    Last revised \@date.\else\large\@date\fi}
  \end{center}
  \begin{center}
    Project web site:\\
    http://www.probabilityandfinance.com
  \end{center}
  \end{titlepage}%
  \setcounter{footnote}{0}%
  \global\let\thanks\relax
  \global\let\maketitle\relax
  \global\let\@thanks\@empty
  \global\let\@author\@empty
  \global\let\@date\@empty
  \global\let\@title\@empty
  \global\let\title\relax
  \global\let\author\relax
  \global\let\date\relax
  \global\let\and\relax
}

\renewenvironment{abstract}{%
  \titlepage\pagenumbering{roman}  
  \null\vfil
  \@beginparpenalty\@lowpenalty
  \begin{center}%
    \Large \bfseries \abstractname
    \@endparpenalty\@M
  \end{center}}%
  {\par\vfill\tableofcontents\thispagestyle{empty}\endtitlepage
  \pagenumbering{arabic}}  

\renewenvironment{thebibliography}[1]
  {\section*{\refname}%
  \addcontentsline{toc}{section}{\refname}
  \@mkboth{\MakeUppercase\refname}{\MakeUppercase\refname}%
  \list{\@biblabel{\@arabic\c@enumiv}}%
    {\settowidth\labelwidth{\@biblabel{#1}}%
    \leftmargin\labelwidth
    \advance\leftmargin\labelsep
    \@openbib@code
    \usecounter{enumiv}%
    \let\p@enumiv\@empty
    \renewcommand\theenumiv{\@arabic\c@enumiv}}%
    \sloppy
    \clubpenalty4000
    \@clubpenalty \clubpenalty
    \widowpenalty4000%
    \sfcode`\.\@m}
    {\def\@noitemerr
    {\@latex@warning{Empty `thebibliography' environment}}%
  \endlist}

\makeatother